\documentclass[runningheads]{llncs}
\usepackage[T1]{fontenc}
\usepackage{graphicx}
\usepackage{amsmath}
\usepackage{amssymb}

\spnewtheorem{observation}[theorem]{Observation}{\bfseries}{\itshape}

\newcommand{\OPT}{\mathrm{OPT}}
\newcommand{\val}{\mathrm{val}}
\newcommand{\ALG}{\mathrm{ALG}}

\begin{document}
\title{Maximum Matching-Match: Hardness and Approximation}

\titlerunning{Maximum Matching-Match: Hardness and Approximation}

\author{Ilie Dumitru\inst{1} \and
Adrian Micl\u au\c s\inst{1} \and
Alexandru Popa\inst{1, 2}}
\authorrunning{Dumitru et al.}

\institute{Department of Computer Science, University of Bucharest, Str. Academiei 14, Bucharest, 010014, Romania \and
National Institute for Research and Development in Informatics, Bulevardul Mareșal Alexandru Averescu 8-10, Bucharest, 011555, Romania}
\maketitle              % typeset the header of the contribution

\begin{abstract}
In this paper, we study \textsc{MaxMMP}, an optimization variant of the Matching-Match Puzzle introduced by Iburi and Uehara (FUN 2024). Given a graph, a partial vertex coloring, and a multiset of colored sticks, the goal is to complete the coloring and assign the sticks to graph edges so as to maximize the number of satisfied edges.

We first prove that \textsc{MaxMMP} is APX-hard by an reduction from \textsc{Max-Cut}. The hardness already holds with two colors, no precolored vertices, and only bichromatic sticks. We then give a simple deterministic $\frac{2}{c(c+1)}$-approximation for completely uncolored instances, improving to $\frac{2}{c(c-1)}$ when all sticks are bichromatic.

Next, we obtain a randomized $\frac{1-\frac{1}{e}}{2c}$-approximation for arbitrary instances with $c$ colors by reducing the remaining coloring choices to monotone submodular maximization under a partition matroid. On bipartite graphs, the approximation ratio improves to $\frac{1-\frac{1}{e}}{c}$. For every fixed $c$, we further obtain deterministic $\frac{1}{2c}$ and $\frac{1}{c}$-approximations on general and bipartite graphs, respectively, in time $n^{O(c^2)}$. Finally, for every fixed number of colors, we show that \textsc{MaxMMP} can be solved exactly in time $n^{O(c^2)}$ on trees and on cographs.
\end{abstract}

\keywords{APX-hardness \and Approximation algorithms \and Submodular maximization \and Exact algorithms}

\section{Introduction}
\label{sec:introduction}

\subsection*{Motivation}

Many complex networks are described through local interaction patterns. Network motifs provide a prominent example: small recurring subgraphs are used to characterize structural organization in biological and other complex networks~\cite{MiloEtAl2002}. When vertices carry types, functions, or roles, the distribution of colors within these patterns may contain additional information, motivating the study of colored motifs~\cite{AdamiEtAl2011}.

Motivated by an inverse version of this viewpoint, we consider the problem of realizing a prescribed multiset of colored pairwise interactions on a fixed graph topology. Instead of starting from a colored network and measuring the interactions that occur, we are given the desired multiplicity of each color pair and seek a vertex coloring that realizes these interactions on the graph edges. This viewpoint coincides with the Matching-Match Puzzle introduced by Iburi and Uehara~\cite{IburiUehara2024}.

Exact realization, however, may be impossible. This naturally leads to an optimization question: how much of the prescribed interaction profile can be realized consistently? We study this question through \textsc{MaxMMP}, the maximization version of the Matching-Match Puzzle.

\subsection*{Informal Problem Definition}

An instance consists of a graph $G=(V,E)$, a set of colors $C=\{1,\ldots,c\}$, a partial coloring of the graph vertices, and exactly $|E|$ sticks, each having a color at each endpoint. A stick of type $\{i,j\}$ represents one desired interaction between colors $i$ and $j$.

In the feasibility problem \textsc{FMMP}, all sticks must be assigned bijectively to graph edges so that the colors at their endpoints agree consistently at every graph vertex. Equivalently, one seeks a vertex coloring whose multiset of edge color pairs coincides exactly with the prescribed multiset of stick types.

In \textsc{MaxMMP}, exact consistency is not required on every edge. We choose a final vertex coloring together with a bijection between sticks and graph edges and maximize the number of edges on which the assigned stick agrees with the colors of both endpoints. Thus, \textsc{MaxMMP} asks for a coloring and assignment that realize as large a portion as possible of the prescribed colored interaction profile.

\subsection*{Previous and Related Work}

Iburi and Uehara introduced the Matching-Match Puzzle and initiated the study of its feasibility problem~\cite{IburiUehara2024}. They proved NP-completeness on paths, cycles, and spiders, and gave polynomial-time algorithms for several restricted graph classes, including complete graphs and stars. Their work concerns exact feasibility. To the best of our knowledge, the optimization variant studied here has not previously been considered.

Our hardness result uses \textsc{Max-Cut}, one of the classical graph optimization problems. In particular, \textsc{Max-Cut} remains APX-complete on cubic graphs~\cite{AlimontiKann2000}. We use an objective-preserving reduction from this restricted case to establish APX-hardness of \textsc{MaxMMP} even with only two colors and a single bichromatic stick type.

Our approximation algorithms are related to monotone submodular maximization under matroid constraints. Fisher, Nemhauser, and Wolsey established the classical $\frac{1}{2}$ guarantee for the greedy algorithm under a matroid constraint~\cite{FisherNemhauserWolsey1978}. C\u{a}linescu, Chekuri, P'al, and Vondr'ak later obtained the $1-\frac{1}{e}$ approximation using the continuous-greedy framework~\cite{CalinescuEtAl2011}. We apply these results after showing that, once the colors on one side of a bipartition are fixed, the remaining coloring choices induce a monotone submodular objective under a partition matroid.

For fixed $c$, we also study a different algorithmic regime in which the remaining coloring problem can be solved exactly by dynamic programming over edge-type histograms. This yields stronger deterministic approximation guarantees at the cost of an $n^{O(c^2)}$ running time, and also leads to exact algorithms on trees and cographs.

\subsection*{Our Results}

We initiate the study of \textsc{MaxMMP} and obtain hardness, approximation, and exact algorithmic results.

\paragraph{APX-hardness.}
We first prove that \textsc{MaxMMP} is APX-hard by an objective-preserving reduction from \textsc{Max-Cut}. The hardness already holds when $c=2$, no graph vertex is precolored, and every stick has type $\{1,2\}$. Under these restrictions, a satisfied stick corresponds exactly to a bichromatic graph edge, so the value of every \textsc{MaxMMP} solution coincides with the size of the corresponding cut. Since \textsc{Max-Cut} is APX-complete on cubic graphs~\cite{AlimontiKann2000}, the hardness persists even when the input graph is cubic.

\paragraph{A simple deterministic approximation.}
We next give an deterministic approximation for completely uncolored instances. By selecting a most frequent stick type and coloring the graph accordingly, we obtain a polynomial-time $\frac{2}{c(c+1)}$-approximation. When every stick is bichromatic, the guarantee improves to $\frac{2}{c(c-1)}$.

\paragraph{Randomized $O(\frac{1}{c})$ approximations.}
For arbitrary precoloring, we give a randomized $\frac{1-\frac{1}{e}}{2c}$-approximation on general graphs. The algorithm fixes the colors on one side of a random bipartition and formulates the remaining choices as monotone submodular maximization under a partition matroid. On bipartite graphs, the random bipartition is unnecessary and the ratio improves to $\frac{1-\frac{1}{e}}{c}$.

\paragraph{Deterministic approximations for fixed $c$.}
For every fixed $c$, we obtain stronger deterministic guarantees in time $n^{O(c^2)}$: a $\frac{1}{2c}$-approximation on general graphs and a $\frac{1}{c}$-approximation on bipartite graphs.

\paragraph{Exact algorithms on trees and cographs.}
Finally, for every fixed $c$, we solve \textsc{MaxMMP} exactly in time $n^{O(c^2)}$ on trees and on cographs, even under arbitrary precoloring.

\section{Preliminaries}
\label{sec:preliminaries}

We consider finite simple undirected graphs $G=(V,E)$, with $n=|V|$ and $m=|E|$. The set of colors is $C=\{1,\ldots,c\}$. An instance contains exactly $m$ sticks, each having one color at each endpoint. We identify a stick with its unordered color type $\{i,j\}$, where $1\leq i\leq j\leq c$, and write $s_{ij}$ for the number of sticks of type $\{i,j\}$.

A partial coloring is a function $\mathcal C_0:V \rightarrow \{0,1,\ldots,c\}$, where color $0$ denotes an initially uncolored vertex. A final coloring $\phi:V\rightarrow C$ extends $\mathcal C_0$ if $\phi(v)=\mathcal C_0(v)$ whenever $\mathcal C_0(v)\neq0$.

For a final coloring $\phi$, let $e_{ij}(\phi)$ denote the number of graph edges whose endpoint colors form the unordered pair $\{i,j\}$.

\begin{problem}[Feasible Matching-Match Puzzle (\textsc{FMMP})]
Given a graph $G=(V,E)$, a collection $S$ of $|E|$ colored sticks using colors from $C=\{1,\ldots,c\}$, and a partial coloring $\mathcal C_0:V\rightarrow \{0,1,\ldots,c\}$, determine whether there exists a final coloring $\phi:V\rightarrow C$ extending $\mathcal C_0$ and a bijection between $E$ and $S$ such that every assigned stick agrees with the colors of the endpoints of its edge.
\end{problem}

Since $|S|=|E|$, feasibility is equivalent to the existence of a coloring $\phi$ extending $\mathcal C_0$ such that $e_{ij}(\phi)=s_{ij}$ for every $1\leq i\leq j\leq c$.

We study the corresponding optimization problem. For a final coloring $\phi$ and a bijection $\pi:E\rightarrow S$, an edge $uv\in E$ is \emph{satisfied} if the stick $\pi(uv)$ can be oriented so that its endpoint colors agree with $\phi(u)$ and $\phi(v)$. Let $\val(\phi,\pi)$ denote the number of satisfied edges.

\begin{problem}[Maximum Matching-Match Puzzle (\textsc{MaxMMP})]
Given a graph $G=(V,E)$, a collection $S$ of $|E|$ colored sticks using colors from $C=\{1,\ldots,c\}$, and a partial coloring $\mathcal C_0:V\rightarrow\{0,1,\ldots,c\}$, find a final coloring $\phi:V\rightarrow C$ extending $\mathcal C_0$ and a bijection $\pi:E\rightarrow S$ maximizing $\val(\phi,\pi)$.
\end{problem}

We write $\OPT$ for the optimum value of a \textsc{MaxMMP} instance. The input is not required to be feasible for \textsc{FMMP}. Otherwise every instance would have optimum $m$.

The following observation will be used throughout the paper.

\begin{lemma}[Histogram lemma]
\label{lem:histogram}
For every final coloring $\phi$ extending $\mathcal C_0$, the maximum number of satisfied edges over all bijections is $\displaystyle\sum_{1\leq i\leq j\leq c}\min{s_{ij},e_{ij}(\phi)}$.
\end{lemma}

\begin{proof}
For each type $\{i,j\}$, at most $s_{ij}$ sticks and at most $e_{ij}(\phi)$ graph edges can be matched correctly, so at most $\min\{s_{ij},e_{ij}(\phi)\}$ edges of that type can be satisfied. This bound is attained independently for every type by pairing as many sticks as possible with graph edges of the same type and assigning the remaining sticks arbitrarily to the remaining edges.
\end{proof}

\section{\textsc{MaxMMP} is APX-hard}
\label{sec:maxcut}

In this section we prove that \textsc{MaxMMP} is APX-hard by a reduction from \textsc{Max-Cut}. In fact, when there are only two colors and every stick is bichromatic, the two optimization problems coincide exactly.

\begin{theorem}
\label{thm:maxcut-exact}
\textsc{MaxMMP} is APX-hard, even when $c=2$, no vertex is precolored, and every stick has type $\{1,2\}$.
\end{theorem}

\begin{proof}
We give an objective-preserving reduction from \textsc{Max-Cut}. Let $G=(V,E)$ be an arbitrary graph and construct a \textsc{MaxMMP} instance on the same graph with colors $1$ and $2$, no precolored vertices, and one stick of type $\{1,2\}$ for every edge of $G$.

Consider any cut $(A,V\setminus A)$ of $G$. Color every vertex of $A$ with color $1$ and every vertex of $V\setminus A$ with color $2$. An edge is bichromatic if and only if it crosses the cut. Since every stick has type $\{1,2\}$, exactly the bichromatic edges can be satisfied. Hence the resulting \textsc{MaxMMP} solution has value equal to the size of the cut.

Conversely, consider any final coloring of the \textsc{MaxMMP} instance. Since only colors $1$ and $2$ are available, the coloring defines the cut $(\phi^{-1}(1),\phi^{-1}(2))$. Every stick has type $\{1,2\}$, so an edge is satisfied if and only if its endpoints receive different colors. Thus the number of satisfied edges is exactly the number of edges crossing the corresponding cut.

Therefore every solution of either problem corresponds to a solution of the other with the same objective value, and $\OPT_{\textsc{MaxMMP}}(G)=\OPT_{\textsc{Max-Cut}}(G)$. The reduction is thus an $L$-reduction with $\alpha=\beta=1$. Since \textsc{Max-Cut} is APX-complete on cubic graphs~\cite{AlimontiKann2000}, the claim follows. Thus, \textsc{MaxMMP} admits no PTAS unless $\mathrm{P}=\mathrm{NP}$.
\end{proof}

\section{A Simple Deterministic Approximation}
\label{sec:elementary-approx}

We next give a simple deterministic approximation algorithm for completely uncolored instances. The idea is to focus only on the most frequent stick type. This gives a polynomial-time approximation ratio of order $\Theta(\frac{1}{c^2})$ and provides a useful baseline for the stronger $\Theta(\frac{1}{c})$-type guarantees developed in the next section.

\begin{theorem}
\label{thm:elementary-approx}
If no vertex is precolored, then \textsc{MaxMMP} admits a polynomial-time $\frac{2}{c(c+1)}$-approximation. If every stick is bichromatic and $c\geq 3$, the approximation ratio improves to $\frac{2}{c(c-1)}$.
\end{theorem}

\begin{proof}
Let $G=(V,E)$ be an instance with $m=|E|$. There are $q=\binom{c+1}{2}$ possible unordered stick types. Since the instance contains exactly $m$ sticks, some type $t$ occurs at least $\frac{m}{q}$ times.

Suppose first that $t=\{i,i\}$ is monochromatic. Color every vertex of $G$ with color $i$. Then every graph edge has type $\{i,i\}$, and therefore all $s_{ii}$ sticks of type $\{i,i\}$ can be satisfied. Hence the resulting solution has value at least $\frac{m}{q}$.

Now suppose that $t=\{i,j\}$ with $i\neq j$. Compute a cut $(A,V\setminus A)$ containing at least $\frac{m}{2}$ edges. Such a cut can be found deterministically by processing the vertices in an arbitrary order and placing each vertex on the side that maximizes the number of crossing edges to already processed vertices. At least half of the edges to previously processed vertices become crossing, and every edge is considered exactly once, so the resulting cut has size at least $\frac{m}{2}$. Color every vertex in $A$ with $i$ and every vertex in $V\setminus A$ with $j$. Every crossing edge then has type $\{i,j\}$. Since $q\geq2$ for $c\geq2$, we have $\frac{m}{q}\leq\frac{m}{2}$. Moreover, $s_{ij}\geq\frac{m}{q}$. Therefore at least $\frac{m}{q}$ edges can be satisfied with sticks of type $\{i,j\}$.

In either case, the algorithm produces a solution of value at least $\frac{m}{q}$. Since $\OPT\leq m$, its approximation ratio is at least $\frac{1}{q}=\frac{2}{c(c+1)}$.

Now assume that every stick is bichromatic. There are only $q=\binom{c}{2}$ possible stick types, so some type $\{i,j\}$ occurs at least $\frac{m}{q}$ times. For $c\geq3$, we have $q\geq3$ and therefore $\frac{m}{q}\leq\frac{m}{2}$. Using a cut with at least $\frac{m}{2}$ crossing edges and coloring its two sides with $i$ and $j$ yields a solution of value at least $\frac{m}{q}$. Hence the approximation ratio is at least $\frac{1}{q}=\frac{2}{c(c-1)}$.
\end{proof}

\begin{remark}
The bichromatic guarantee $\frac{2}{c(c-1)}$ applies only for $c\geq3$. When $c=2$, every stick has type $\{1,2\}$, and the problem is exactly \textsc{Max-Cut} by Theorem~\ref{thm:maxcut-exact}.
\end{remark}

The guarantee of Theorem~\ref{thm:elementary-approx} is only of order $O(\frac{1}{c^2})$. In the next section, we improve this dependence to order $O(\frac{1}{c})$ by exploiting monotone submodular maximization.

\section{$\Theta(\frac{1}{c})$ Approximations}
\label{sec:submodular-approx}

We now improve the $\Theta(\frac{1}{c^2})$ guarantee of Theorem~\ref{thm:elementary-approx} to a guarantee of order $\Theta(\frac{1}{c})$. Unlike the elementary approximation, the results in this section allow arbitrary precoloring.

The main idea is to fix the colors on one side of a bipartition and optimize the colors on the other side. Once the first side is fixed, the contribution of the remaining coloring choices can be expressed as a monotone submodular function subject to a partition matroid constraint. We then apply the $1-\frac{1}{e}$ approximation of C\u{a}linescu, Chekuri, P'al, and Vondr'ak~\cite{CalinescuEtAl2011}.

\begin{theorem}
\label{thm:general-approx}
\textsc{MaxMMP} admits a randomized $\frac{1-\frac{1}{e}}{2c}$-approximation on general graphs, even with arbitrary precoloring.
\end{theorem}

\begin{proof}
Fix an optimal solution $(\phi^*,\pi^*)$ of value $\OPT$.

Independently place every vertex into one of two sets $L$ and $R$ with probability $\frac{1}{2}$. Every precolored vertex in $L$ keeps its prescribed color, while every unprecolored vertex in $L$ receives a color chosen independently and uniformly from $C$. We ignore edges with both endpoints in $L$ or both endpoints in $R$ and optimize only over the crossing edges.

For every vertex $v\in R$ and every color $a$ allowed at $v$, create an option $(v,a)$. If $v$ is precolored, only its prescribed color is allowed. The options are partitioned by vertex, and an independent set contains at most one color option for each vertex. Thus the feasible choices form a partition matroid.

For every stick type $t$ and option $(v,a)$, let $w_{v,a,t}$ be the number of neighbors $u\in L$ for which the crossing edge $uv$ has type $t$ when $v$ receives color $a$. For a set $A$ of selected options, define
$h_t(A)=\displaystyle\sum_{(v,a)\in A} w_{v,a,t}$
and
$F(A)=\displaystyle\sum_t \min\{s_t,h_t(A)\}$.

Each function $h_t$ is nonnegative and modular. Hence $A\mapsto\min\{s_t,h_t(A)\}$ is monotone and submodular, and therefore so is $F$.

If $A$ contains one color choice for every vertex of $R$, then $h_t(A)$ is exactly the number of crossing edges of type $t$. By Lemma~\ref{lem:histogram}, $F(A)$ is the maximum number of crossing edges that can be satisfied under the resulting coloring. Since $F$ is monotone, a partial independent set can be completed by assigning arbitrary allowed colors to the remaining vertices without decreasing its value.

We next compare the optimum of this submodular problem with the original optimum. Call an edge satisfied by $(\phi^*,\pi^*)$ \emph{surviving} if it crosses the random partition and its endpoint in $L$ receives the same color as under $\phi^*$.

Consider an edge satisfied by the optimum. It crosses $(L,R)$ with probability $\frac{1}{2}$. Conditioned on crossing, if its endpoint in $L$ was initially uncolored, it receives its color under $\phi^*$ with probability $\frac{1}{c}$. If it was precolored, its prescribed color already agrees with $\phi^*$. Thus every edge satisfied by the optimum survives with probability at least $\frac{1}{2c}$.

Let $Z$ be the number of surviving optimal edges. By linearity of expectation,
$\mathbb{E}[Z]\geq\frac{\OPT}{2c}$.

Now fix the random partition and the colors assigned to $L$, and color every vertex of $R$ according to $\phi^*$. For every stick type $t$, let $Z_t$ be the number of surviving optimal edges of type $t$. Since all these edges are satisfied in the global optimum, $Z_t\leq s_t$. Moreover, under the coloring induced by $\phi^*$, there are at least $Z_t$ crossing edges of type $t$. Hence the corresponding second-stage solution has value at least
$\displaystyle \displaystyle\sum_t Z_t=Z$.

Therefore the optimum value of the submodular instance is at least $Z$. Applying the $1-\frac{1}{e}$ approximation for monotone submodular maximization under a matroid constraint gives, conditioned on the first-stage random choices, expected value at least $(1-\frac{1}{e})Z$. Taking expectation over the first stage,
$\mathbb{E}[\ALG]\geq(1-\frac{1}{e})\mathbb{E}[Z]\geq\frac{1-\frac{1}{e}}{2c}\OPT$. Thus, the algorithm is a randomized $\frac{1-\frac{1}{e}}{2c}$-approximation.
\end{proof}

The factor $\frac{1}{2}$ in Theorem~\ref{thm:general-approx} arises only from the probability that an edge crosses the random bipartition. For bipartite graphs, the given bipartition separates every edge, and this loss disappears.

\begin{corollary}
\label{cor:bipartite-approx}
On bipartite graphs, \textsc{MaxMMP} admits a randomized $\frac{1-\frac{1}{e}}{c}$ approximation, even with arbitrary precoloring.
\end{corollary}

\begin{proof}
Let $(L,R)$ be a fixed bipartition of the input graph and apply the construction from the proof of Theorem~\ref{thm:general-approx}, without the random partition step. Every unprecolored vertex in $L$ receives an independent uniformly random color, while precolored vertices keep their prescribed colors.

For every edge satisfied by a fixed optimum, its endpoint in $L$ receives its optimal color with probability at least $\frac{1}{c}$. Hence the expected number $Z$ of surviving optimal edges satisfies $\mathbb{E}[Z]\geq\frac{\OPT}{c}$. The same submodular argument as in Theorem~\ref{thm:general-approx} shows that the optimum second-stage value is at least $Z$. Applying the $1-\frac{1}{e}$ approximation therefore gives
$\mathbb{E}[\ALG]\geq\frac{1-\frac{1}{e}}{c}\OPT$.
\end{proof}

\section{Deterministic $\Theta(\frac{1}{c})$ Approximations for Fixed $c$}
\label{sec:fixed-c-approx}

The randomized algorithms of Section~\ref{sec:submodular-approx} run in polynomial time even when the number of colors is part of the input. We now consider the regime in which $c$ is fixed. In this case, once the coloring on one side of a cut is fixed, the remaining coloring choices can be optimized exactly over the crossing edges by dynamic programming. This yields deterministic approximation ratios $\frac{1}{2c}$ on general graphs and $\frac{1}{c}$ on bipartite graphs.

We first establish the completion procedure used by both algorithms.

\begin{lemma}
\label{lem:fixed-side-completion}
Let $V=L\cup R$ be a partition of the vertices of a \textsc{MaxMMP} instance, and suppose that the colors of all vertices in $L$ are fixed consistently with the precoloring. For fixed $c$, one can choose colors for the vertices in $R$ so as to maximize the number of satisfiable crossing edges in time $n^{O(c^2)}$.
\end{lemma}

\begin{proof}
Let $q=\binom{c+1}{2}$ be the number of stick types. For every vertex $v\in R$ and every color $a$ allowed at $v$, define a vector $w(v,a)\in\mathbb{N}^q$, where the coordinate corresponding to type $t$ is the number of neighbors $u\in L$ for which $uv$ has type $t$ when $v$ receives color $a$.

Since the coloring of $L$ is fixed, every crossing edge is incident with exactly one vertex of $R$. Hence, if the vertices of $R$ receive colors $a_v$, the histogram of the crossing edges is exactly
$\displaystyle\sum_{v\in R}w(v,a_v)$.

We process the vertices of $R$ one by one and maintain all attainable crossing-edge histograms. Initially only the zero histogram is attainable. When processing a vertex $v$, for every currently attainable histogram $h$ and every color $a$ allowed at $v$, we add the histogram $h+w(v,a)$. If $v$ is precolored, only its prescribed color is considered.

Every coordinate of a histogram is at most $m$, so there are at most $(m+1)^q$ states. Since $q=\binom{c+1}{2}=O(c^2)$, the algorithm runs in time $n^{O(c^2)}$ for fixed $c$.

For every attainable histogram $h$, Lemma~\ref{lem:histogram} shows that the maximum number of crossing edges that can be satisfied is $\displaystyle\sum_t\min\{s_t,h_t\}$. Taking the maximum over all attainable histograms therefore gives an optimal coloring of $R$ with respect to the crossing edges.
\end{proof}

We use this lemma together with a small deterministic family of cuts. Let $r=\lceil\log_2 n\rceil$ and assign every vertex $v$ a distinct vector $x_v\in \{0,1\}^r$. For every $p\in{0,1}^r$, define
$L_p={v:\langle p,x_v\rangle=0}$ and $R_p=V\setminus L_p$, where the inner product is over $\mathbb{F}_2$. The family contains $2^r\leq2n$ cuts. Moreover, every pair of distinct vertices is separated by exactly half of them: for $u\neq v$, the vector $x_u+x_v$ is nonzero, and exactly half of the vectors $p$ satisfy $\langle p,x_u+x_v\rangle=1$.

\begin{theorem}
\label{thm:fixed-c-deterministic}
For every fixed $c$, \textsc{MaxMMP} admits a deterministic $\frac{1}{2c}$ approximation on general graphs in time $n^{O(c^2)}$, even with arbitrary precoloring.
\end{theorem}

\begin{proof}
Consider every cut $(L_p,R_p)$ in the family above and every color $a\in C$. For the candidate $(p,a)$, keep the prescribed colors of all precolored vertices in $L_p$ and color every unprecolored vertex in $L_p$ with $a$. Then use Lemma~\ref{lem:fixed-side-completion} to choose the colors of $R_p$ optimally with respect to the crossing edges. Return the best solution over all choices of $p$ and $a$.

Fix an optimal solution $(\phi^*,\pi^*)$ of value $\OPT$, and let $S^*$ be the set of edges satisfied by this solution. For a candidate $(p,a)$, let $S_{p,a}\subseteq S^*$ contain those edges that cross $(L_p,R_p)$ and whose endpoint in $L_p$ receives the same color under the candidate as under $\phi^*$.

For every edge in $S_{p,a}$, coloring its endpoint in $R_p$ according to $\phi^*$ restores exactly the endpoint colors that it has in the optimal solution. Hence all edges of $S_{p,a}$ can simultaneously be satisfied using the sticks assigned to them by $\pi^*$. Since Lemma~\ref{lem:fixed-side-completion} finds an optimal completion of $R_p$, the value obtained for candidate $(p,a)$ is at least $|S_{p,a}|$.

Now fix an edge $uv\in S^*$. The endpoints $u$ and $v$ are separated by exactly half of the cuts in the family. For each such cut, consider the endpoint lying in $L_p$. If that endpoint is precolored, its color necessarily agrees with $\phi^*$ for every choice of $a$. If it is unprecolored, exactly one of the $c$ choices of $a$ agrees with its color under $\phi^*$. Thus, over all pairs $(p,a)$, the edge belongs to $S_{p,a}$ for at least a $\frac{1}{2c}$ fraction of the candidates.

Averaging over all optimal satisfied edges gives
$\frac{1}{2^r c}\displaystyle\sum_{p,a}|S_{p,a}|\geq\frac{\OPT}{2c}$.
Therefore some candidate satisfies $|S_{p,a}|\geq\frac{\OPT}{2c}$, and the solution returned by the algorithm has value at least $\frac{\OPT}{2c}$.

There are at most $2nc$ candidates, and each is solved in time $n^{O(c^2)}$ by Lemma~\ref{lem:fixed-side-completion}. The total running time is therefore $n^{O(c^2)}$ for fixed $c$.
\end{proof}

On bipartite graphs, the family of cuts is unnecessary because the given bipartition already separates every edge.

\begin{corollary}
\label{cor:fixed-c-bipartite}
For every fixed $c$, \textsc{MaxMMP} admits a deterministic $\frac{1}{c}$ approximation on bipartite graphs in time $n^{O(c^2)}$, even with arbitrary precoloring.
\end{corollary}

\begin{proof}
Let $(L,R)$ be a bipartition of the input graph. For every color $a\in C$, keep the prescribed colors of the precolored vertices in $L$, color every unprecolored vertex in $L$ with $a$, and apply Lemma~\ref{lem:fixed-side-completion} to optimize the coloring of $R$.

Fix an optimal solution and its set $S^*$ of satisfied edges. For each $a\in C$, let $S_a\subseteq S^*$ contain the edges whose endpoint in $L$ agrees with its color in the optimum after the above coloring of $L$. If this endpoint is precolored, the edge belongs to every $S_a$; otherwise, it belongs to exactly one $S_a$. Consequently,
$\displaystyle\sum_{a\in C}|S_a|\geq\OPT$,
and hence some color $a$ satisfies $|S_a|\geq\frac{\OPT}{c}$.

Coloring $R$ according to the global optimum would satisfy all edges in $S_a$, so the optimal completion returned by Lemma~\ref{lem:fixed-side-completion} has value at least $|S_a|$. Taking the best of the $c$ candidates therefore gives a deterministic $\frac{1}{c}$-approximation.
\end{proof}

Thus, for fixed $c$, exact optimization of the second-stage coloring improves the randomized ratios of Section~\ref{sec:submodular-approx} while also removing randomization.

\section{Exact Algorithms on Trees and Cographs for Fixed $c$}
\label{sec:exact-structured}

We now turn from approximation to exact optimization on structured graph classes. For every fixed number of colors $c$, \textsc{MaxMMP} can be solved exactly in time $n^{O(c^2)}$ on trees and on cographs, even with arbitrary precoloring.

Throughout this section, let
$q=\binom{c+1}{2}$
denote the number of unordered stick types. For a coloring of a subgraph, we represent the multiplicities of its edge types by a histogram $h\in\{0,\ldots,m\}^q$.

\subsection{Trees}

Let $T$ be a tree rooted at an arbitrary vertex $r$. For a vertex $v$, let $T_v$ denote the subtree rooted at $v$.

For every vertex $v$, color $a\in C$ allowed at $v$, and histogram $h$, define a state
$D_v(a,h)$
which is true if and only if $T_v$ admits a coloring extending the precoloring such that $v$ receives color $a$ and the edges of $T_v$ realize exactly the histogram $h$.

\begin{theorem}
\label{thm:tree-exact}
For every fixed $c$, \textsc{MaxMMP} can be solved exactly on trees in time $n^{O(c^2)}$, even with arbitrary precoloring.
\end{theorem}

\begin{proof}
We compute the states $D_v(a,h)$ bottom-up.

If $v$ is a leaf, then the only realizable histogram is the zero histogram. Thus $D_v(a,0)$ is true for every color $a$ allowed at $v$.

Now let $v$ have children $u_1,\ldots,u_d$, and fix an allowed color $a$ for $v$. We combine the children one at a time. Suppose that for a child $u$, the state $D_u(b,h')$ is true. If $u$ receives color $b$, then the edge $vu$ contributes one additional edge of type ${a,b}$. Hence this child contributes the histogram
$h'+\mathbf{e}*{{a,b}}$,
where $\mathbf{e}*{{a,b}}$ is the unit vector corresponding to type ${a,b}$.

Starting with the zero histogram, we repeatedly add one realizable contribution from each child. After all children have been processed, the resulting histograms are exactly those realizable in $T_v$ with $v$ colored $a$.

Each histogram has $q$ coordinates, each between $0$ and $m$, so there are at most $(m+1)^q$ possible histograms. Since there are $c$ possible colors for the root of each subtree, the number of states per vertex is at most $c(m+1)^q$. As $q=\binom{c+1}{2}=O(c^2)$, the algorithm runs in time $n^{O(c^2)}$ for fixed $c$.

At the root $r$, every realizable coloring of the whole tree is represented by some state $D_r(a,h)$. By Lemma~\ref{lem:histogram}, the maximum number of satisfied edges for such a coloring is
$\displaystyle\sum_t\min\{s_t,h_t\}$.
Taking the maximum over all true states $D_r(a,h)$ gives the optimum value of the instance.
\end{proof}

\subsection{Cographs}

A cograph can be constructed recursively from single vertices by disjoint union and complete join. We use a cotree representing this recursive construction.

For a node $x$ of the cotree, let $G_x$ denote the subgraph represented by the subtree rooted at $x$. In addition to the edge-type histogram, we must record the number of vertices of each color, because at a join node all edges between the two child subgraphs are introduced.

For a vector $\alpha=(\alpha_1,\ldots,\alpha_c)$ and a histogram $h$, define
$D_x(\alpha,h)$
to be true if and only if $G_x$ admits a coloring extending the precoloring such that exactly $\alpha_i$ vertices receive color $i$ and the edges of $G_x$ realize exactly the histogram $h$.

\begin{theorem}
\label{thm:cograph-exact}
For every fixed $c$, \textsc{MaxMMP} can be solved exactly on cographs in time $n^{O(c^2)}$, even with arbitrary precoloring.
\end{theorem}

\begin{proof}
We compute the states bottom-up over a cotree.

At a leaf corresponding to a vertex $v$, for every color $i$ allowed at $v$, the state with $\alpha_i=1$, all other entries of $\alpha$ equal to $0$, and zero edge histogram is realizable.

Consider an internal node $x$ with children $y$ and $z$.

If $x$ is a disjoint-union node, no edges are added between $G_y$ and $G_z$. Therefore, if $D_y(\alpha,h)$ and $D_z(\beta,h')$ are true, then
$D_x(\alpha+\beta,h+h')$
is true.

Now suppose that $x$ is a join node. Then every vertex of $G_y$ is adjacent to every vertex of $G_z$. Again combine states $D_y(\alpha,h)$ and $D_z(\beta,h')$. The new color-count vector is $\alpha+\beta$. The internal edge histograms contribute $h+h'$, and the join edges contribute a histogram $g$ determined entirely by $\alpha$ and $\beta$.

For every color $i$, the number of new edges of type $\{i,i\}$ is
$g_{ii}=\alpha_i\beta_i$.
For distinct colors $i<j$, the number of new edges of type $\{i,j\}$ is
$g_{ij}=\alpha_i\beta_j+\alpha_j\beta_i$.
Hence the resulting histogram is $h+h'+g$.

There are at most $(n+1)^c$ possible color-count vectors and at most $(m+1)^q$ possible histograms. Thus the number of states at each cotree node is at most
$(n+1)^c(m+1)^q=n^{O(c^2)}$
for fixed $c$. Combining pairs of states still gives an overall running time of $n^{O(c^2)}$.

At the root of the cotree, each realizable coloring of the whole graph corresponds to some state $D_x(\alpha,h)$. By Lemma~\ref{lem:histogram}, its optimum stick assignment satisfies
$\displaystyle\sum_t\min\{s_t,h_t\}$
edges. Maximizing this quantity over all realizable root states yields the optimum value of the \textsc{MaxMMP} instance.
\end{proof}

As an immediate consequence, for every fixed $c$, \textsc{FMMP} can be solved in time $n^{O(c^2)}$ on trees and cographs, even with arbitrary precoloring, since feasibility is equivalent to realizing exactly the prescribed stick histogram.

\section{Conclusion}
\label{sec:conclusion}

We initiated the study of \textsc{MaxMMP}, the optimization variant of the Matching-Match Puzzle. We proved APX-hardness even with two colors, no precolored vertices, and only bichromatic sticks, via an objective-preserving reduction from \textsc{Max-Cut}.

On the algorithmic side, we gave a simple deterministic $\frac{2}{c(c+1)}$-approximation for completely uncolored instances and randomized $\Theta(\frac{1}{c})$-approximations for arbitrary precoloring: $\frac{1-\frac{1}{e}}{2c}$ on general graphs and $\frac{1-\frac{1}{e}}{c}$ on bipartite graphs. For fixed $c$, we further obtained deterministic $\frac{1}{2c}$- and $\frac{1}{c}$-approximations on general and bipartite graphs, respectively, in time $n^{O(c^2)}$. We also gave exact $n^{O(c^2)}$-time algorithms on trees and cographs, which immediately yield the same running time for \textsc{FMMP} on these classes.

A main open question is whether the $\Theta(\frac{1}{c})$ dependence is necessary, or whether a constant-factor approximation independent of $c$ is possible. It would also be interesting to determine tighter bounds for small fixed $c$ and to extend the exact algorithms to broader graph classes.

\bibliographystyle{splncs04}
\bibliography{bibl_cocoa}

\end{document}